\documentclass[a4paper,onecolumn,11pt,accepted=2022-06-20]{quantumarticle}
\pdfoutput=1

\usepackage{xcolor}
\definecolor{mygreen}{RGB}{80,180,0}
\usepackage{ifdraft}
\ifdraft{\newcommand{\authnote}[3]{{\color{#3} {\bf  #1:} #2}}}{\newcommand{\authnote}[3]{}}

\newcommand{\Ewin}[1]{\authnote{[Ewin}{#1{\bf]}}{red}}

\newif\ifcount
\ifdraft{\counttrue}

\usepackage{algorithm}
\usepackage{algcompatible}
\usepackage{amssymb}
\usepackage{amsthm} 
\usepackage[english]{babel}
\usepackage{bbm}
\usepackage{enumitem}
\usepackage{float}
\usepackage[T1]{fontenc}
\usepackage[utf8]{inputenc}
\usepackage{mathtools} 
\usepackage{mleftright}\mleftright
\usepackage{multirow}
\usepackage[numbers]{natbib}
\usepackage{thm-restate}
\usepackage{thmtools} 
\usepackage{tikz}
\usepackage{tikz-cd}
\usepackage{url}
\usepackage{verbatim}
\usepackage{wrapfig}

\makeatletter
\def\namedlabel#1#2{\begingroup
    #2%
    \def\@currentlabel{#2}%
    \phantomsection\label{#1}\endgroup
}

\usepackage[final=true,colorlinks = true]{hyperref}
\usepackage[capitalize]{cleveref}

\newcommand{\BB}[1]{\mathbb{#1}}

\newcommand{\CC}[1]{\mathcal{#1}}

\newcommand{\eps}{\varepsilon}

\DeclareMathOperator*{\E}{{\mathbb{E}}}
\DeclareMathOperator*{\Var}{{\mathbb{V}}}
\DeclareMathOperator{\nnz}{nnz}

\DeclareMathOperator{\poly}{poly}

\DeclareMathOperator{\rank}{rank}

\newcommand{\bigO}[1]{\mathcal{O}\left( #1 \right)}
\newcommand{\bigOt}[1]{\widetilde{\mathcal{O}}\left( #1 \right)}
\newcommand{\tsample}{\mathcal{T}_{\textsc{sample}}}

\newcommand{\C}{\BB{C}}
\newcommand{\R}{\BB{R}}
\DeclareMathOperator{\SQ}{SQ}
\DeclareMathOperator{\Q}{Q}
\newcommand{\fr}{{\mathrm{F}}}

\DeclarePairedDelimiter{\nrm}{\lVert}{\rVert}

\DeclarePairedDelimiter{\abs}{\lvert}{\rvert}
\DeclarePairedDelimiter{\norm}{\lVert}{\rVert}

\DeclarePairedDelimiter\ket{\lvert}{\rangle}
\DeclarePairedDelimiterX\braket[2]{\langle}{\rangle}{#1 \delimsize\vert #2}
\newcommand{\underflow}[2]{\underset{\kern-60mm \overbrace{#1} \kern-60mm}{#2}}

\declaretheorem[numberwithin=section]{theorem}
\declaretheorem[sibling=theorem]{lemma}

\declaretheorem[sibling=theorem,name=Proposition]{proposition}
\declaretheorem[numbered=no,style=definition,name=Problem]{problem}

\declaretheorem[sibling=theorem]{claim}
\declaretheorem[sibling=theorem]{definition}
\declaretheorem[style=definition,sibling=theorem]{fact}
\declaretheorem[numbered=no,style=definition,name=Fact]{fact*}

\newtheorem*{theorem*}{Theorem}
\newtheorem*{corollary*}{Corollary}
\newtheorem*{proposition*}{Proposition}
\newtheorem*{lemma*}{Lemma}
\newtheorem*{claim*}{Claim}
\newtheorem*{problem*}{Problem}
\newtheorem*{definition*}{Definition}

\allowdisplaybreaks

\usetikzlibrary{backgrounds,fit,decorations.pathreplacing,calc}

\begin{document}
\title{An improved quantum-inspired algorithm for linear regression}

\author{András Gilyén}
\email{gilyen@renyi.hu}
\affiliation{Alfréd Rényi Institute of Mathematics}
\thanks{Formerly at the Institute for Quantum Information and Matter, California Institute of Technology.}
\author{Zhao Song}
\email{zsong@adobe.com}
\affiliation{Adobe Research}
\author{Ewin Tang}
\affiliation{University of Washington}
\email{ewint@cs.washington.edu}
\orcid{0000-0002-7451-9687}

\maketitle

\begin{abstract}
    We give a classical algorithm for linear regression analogous to the quantum matrix inversion algorithm [Harrow, Hassidim, and Lloyd, Physical Review Letters'09] for low-rank matrices [Wossnig, Zhao, and Prakash, Physical Review Letters'18], when the input matrix $A$ is stored in a data structure applicable for QRAM-based state preparation.

Namely, suppose we are given an $A \in \mathbb{C}^{m\times n}$ with minimum non-zero singular value $\sigma$ which supports certain efficient $\ell_2$-norm importance sampling queries, along with a $b \in \mathbb{C}^m$.
Then, for some $x \in \mathbb{C}^n$ satisfying $\|x - A^+b\| \leq \varepsilon\|A^+b\|$, we can output a measurement of $\ket{x}$ in the computational basis and output an entry of $x$ with classical algorithms that run in $\tilde{\mathcal{O}}\big(\frac{\|A\|_{\mathrm{F}}^6\|A\|^6}{\sigma^{12}\varepsilon^4}\big)$ and $\tilde{\mathcal{O}}\big(\frac{\|A\|_{\mathrm{F}}^6\|A\|^2}{\sigma^8\varepsilon^4}\big)$ time, respectively.
This improves on previous ``quantum-inspired'' algorithms in this line of research by at least a factor of $\frac{\|A\|^{16}}{\sigma^{16}\varepsilon^2}$ [Chia, Gily\'{e}n, Li, Lin, Tang, and Wang, STOC'20].
As a consequence, we show that quantum computers can achieve at most a factor-of-12 speedup for linear regression in this QRAM data structure setting and related settings.
Our work applies techniques from sketching algorithms and optimization to the quantum-inspired literature. 
Unlike earlier works, this is a promising avenue that could lead to feasible implementations of classical regression in a quantum-inspired settings, for comparison against future quantum computers.
\end{abstract}
\newpage


\section{Introduction}
An important question for the future of quantum computing is whether we can use quantum computers to speed up machine learning~\cite{preskill2018QuantCompNISQEra}.
Answering this question is a topic of active research~\cite{childs09,aaronson15read,biamonte2016QuantumMachineLearning}.
One potential avenue to an affirmative answer is through quantum linear algebra algorithms, which can perform linear regression~\cite{harrow2009QLinSysSolver} (and compute similar linear algebra expressions~\cite{gilyen2018QSingValTransf}) in time poly-logarithmic in input dimension, with some restrictive ``quantum'' assumptions~\cite{childs09,aaronson15read}.
For machine learning tasks, a natural setting of these assumptions is that the input is low-rank, in which case the above quantum algorithms require that the input is given in a manner allowing efficient quantum state preparation.
This state preparation assumption is typically satisfied via a data structure instantiated with quantum random access memory (QRAM)\footnote{In this paper, we always use QRAM in combination with a data structure used for efficiently preparing states $\ket{0} \to \sum_i \frac{x_i}{\|x\|}\ket{i}$ corresponding to vectors $x\in\C^n$.}~\cite{giovannetti2007QuantumRAM,prakash2014QLinAlgAndMLThesis}.
A recent paper shows that, by exploiting these assumptions, classical algorithms can also perform quantum linear algebra with only a polynomial slowdown~\cite{chia2019SampdSubLinLowRankFramework}, meaning that here quantum computers do not give the exponential speedup in dimension that one might hope for.
However, these algorithms have running times with high exponents, leaving open the possibility that quantum computers admit large polynomial speedups for these problems.
For example, given our current understanding, the quantum recommendation systems algorithm~\cite{kerenidis2016QRecSys}, with running time $\mathcal{O}^*\!\big(\frac{\|A\|_\fr}{\sigma}\big)$, could still give a large polynomial quantum speedup, since the best-known classical algorithm has running time $\tilde{\mathcal{O}}^*\!\big(\frac{\|A\|_\fr^6\|A\|^{16}}{\sigma^{24}}\big)$~\cite{chia2019SampdSubLinLowRankFramework} (improving the original running time of $\tilde{\mathcal{O}}^*\!\big(\big(\frac{\|A\|_\fr}{\sigma}\big)^{\!24}\big)$~\cite{tang2018QuantumInspiredRecommSys}).\footnote{We define $\tilde{\mathcal{O}}(T)$ as $\mathcal{O}(T\cdot \mathrm{polylog}(T))$, and define $\mathcal{O}^*$ to be big $\mathcal{O}$ notation, hiding polynomial dependence on $\eps$ and poly-logarithmic dependence on dimension.}
It is an open question \emph{whether any quantum linear algebra algorithm admits a large polynomial speedup}, compared to classical numerical linear algebra~\cite{preskill2018QuantCompNISQEra,kllp19,kerenidis2017QGradDesc}.

We focus on this open question for the problem of low-rank linear regression, where we are given a matrix $A \in \C^{m\times n}$ with minimum singular value $\sigma$ and a vector $b \in \C^m$, and asked to approximate $x^* := \arg\min_x \frac{1}{2}\|Ax - b\|_2^2$.
Linear regression has been an influential primitive for quantum machine learning (QML) \cite{biamonte2016QuantumMachineLearning} since the invention of Harrow, Hassidim, and Lloyd's algorithm (HHL) originally introduced for sparse $A$~\cite{harrow2009QLinSysSolver,dhmsuw18}.
The low-rank variant that we consider commonly appears in the QML literature \cite{wossnig2018QLinSysAlgForDensMat,rebentrost2014QSVM}.
These algorithms, like all the algorithms we consider, are called ``low-rank'' because their runtimes depend on $\frac{\|A\|_\fr^2}{\sigma^2} \geq \rank A$, so in some sense, good performance requires the matrix $A$ to be strictly low-rank.
However, this restriction can be relaxed with regularization, which replaces $\sigma^2$ by $\sigma^2 + \lambda$, so that for reasonable choices of the regularization parameter $\lambda$, $A$ only needs to be approximately low-rank.
The current state-of-the-art quantum algorithm can produce a state $\eps$-close to $\ket{A^+b}$ in $\ell_2$-norm in $\mathcal{O}^*\!\big(\frac{\|A\|_\fr}{\|A\|}\frac{\|A\|}{\sigma}\big)$ time~\cite{chakraborty2018BlockMatrixPowers}, given $A$ and $b$ in the aforementioned data structure.
We think about this running time as depending polynomially on the (square root of) stable rank $\frac{\|A\|_\fr}{\|A\|}$ and the condition number $\frac{\|A\|}{\sigma}$.
Note that being able to produce quantum states $\ket{x} := \frac{1}{\|x\|}\sum_i x_i\ket{i}$ corresponding to a desired vector $x$ is akin to a classical sampling problem, and is different from outputting $x$ itself.
Though recent ``quantum-inspired'' classical algorithms~\cite{chia2020QuantInsLinEqSolving} demonstrate that this quantum algorithm does not admit exponential speedups, the best previous classical algorithm of Chia, Gily\'{e}n, Li, Lin, Tang and Wang \cite{chia2019SampdSubLinLowRankFramework} runs in $\tilde{\mathcal{O}}^*\!\big(\big(\frac{\|A\|_\fr}{\|A\|}\big)^{\!6}\big(\frac{\|A\|}{\sigma}\big)^{\!28}\big)$ time, leaving significant room for \emph{polynomial} quantum speedups.\footnote{The $\eps$ dependence for the quantum algorithm is $\log\frac{1}{\eps}$, compared to the classical algorithm which gets $\frac{1}{\eps^6}$. While this suggests an exponential speedup in $\eps$, it appears to only hold for sampling problems, such as measuring $\ket{A^+b}$ in the computational basis. Learning information from the output quantum state $\ket{A^+b}$ generally requires $\poly(\frac{1}{\eps})$ samples, preventing exponential separation for computational problems.}

Our main result tightens this gap, giving an algorithm running in
$\mathcal{O}^*\!\big(\big(\frac{\|A\|_\fr}{\|A\|}\big)^{\!6}\big(\frac{\|A\|}{\sigma}\big)^{\!12}\big)$ time.
Roughly, this tightens the separation between quantum and classical from 1-to-28 to 1-to-12.
As a bonus, this algorithm is a factor of $\frac{\|A\|^4}{\sigma^4}$ faster if we want to compute an entry of the output, rather than to sample an index.
Our algorithm is a carefully-analyzed instance of stochastic gradient descent that exploits the sampling access provided in our setting, combined with a technical sketching step to enforce sparsity of the input vector $b$.
Since it is an iterative algorithm, it is potentially more practical to implement and use as a benchmark against future scalable quantum computers~\cite{adbl20}.
Our result suggests that other tasks that can be solved through iterative methods, like primal-dual algorithms, may have quantum-classical gaps that are smaller than existing quantum-inspired results suggest.

\subsection{Our model}
We now introduce the input and output model that we and all prior quantum-inspired algorithms in this line of research use~\cite{tang2018QInspiredClassAlgPCA,chia2019SampdSubLinLowRankFramework}.
The motivation for this model is to be a classical analogue to the input model of QML algorithms: many such algorithms assume that input is stored in a particular data structure in QRAM \cite{wossnig2018QLinSysAlgForDensMat,rebentrost2014QSVM,prakash2014QLinAlgAndMLThesis,biamonte2016QuantumMachineLearning,preskill2018QuantCompNISQEra,chiprsw18}, including current QML algorithms for low-rank linear regression (our focus)~\cite{chakraborty2018BlockMatrixPowers}.\footnote{Although current QRAM proposals suggest that quantum hardware implementing QRAM may be realizable with essentially only logarithmic overhead in the running time~\cite{giovannetti2007QuantumRAM}, an actual physical implementation would require substantial advances in quantum technology in order to maintain coherence for a long enough time~\cite{arunachalam2015RobustnessBuckBrigQRAM}.}
So, our comparison classical algorithms must also assume that input is stored in this ``quantum-inspired'' data structure, which supports several fully-classical operations.
We first define the quantum-inspired data structure for vectors (Definition~\ref{def:lengthsqvector}), then for matrices (Definition~\ref{def:lengthsqmat}).
\begin{definition}[Vector-based data-structure, $\SQ(v)$ and $\Q(v)$]\label{def:lengthsqvector}
    For any vector $v\in \C^n$, let $\SQ(v)$ denote a data-structure that supports the following operations:
    \begin{enumerate}
        \item \textsc{Sample}$()$, which outputs the entry $i$ with probability $|v_i|^2 / \| v \|^2$. 
        \item \textsc{Query}$(i)$, which takes $i \in [n]$ as input and outputs $v_i$. 
        \item \textsc{Norm}$()$, which outputs $\| v \|$.\footnote{When we claim that we can call \textsc{Norm}() on output vectors, we will mean that we can output a constant approximation: a number in $[0.9\|v\|, 1.1\|v\|]$.}
    \end{enumerate}
    Let $\mathcal{T}(v)$ denote the max time it takes for the data structure to respond to any query.
    If we only allow the \textsc{Query} operation, the data-structure is called $\Q(v)$.
\end{definition}
Notice that the \textsc{Sample} function is the classical analogue of quantum state preparation, since the distribution being sampled is identical to the one attained by measuring $\ket{v}$ in the computational basis.
We now define the quantum-inspired data structure for matrices.
\begin{definition}[Matrix-based data-structure, $\SQ(A)$]\label{def:lengthsqmat}
    For any matrix $A \in \C^{m \times n}$, let $\SQ(A)$ denote a data-structure that supports the following operations:
    \begin{enumerate}
        \item \textsc{Sample1}$()$, which outputs $i \in [m]$ with probability $\| A_{i,*} \|^2 / \| A \|_\fr^2$. 
        \item \textsc{Sample2}$(i)$, which takes a row index $i \in [m]$ as input and outputs the column index $j \in [n]$ with probability $| A_{i,j} |^2 / \| A_{i,*} \|^2$. 
        \item \textsc{Query}$(i,j)$, which takes $i \in [m]$ and $j \in [n]$ as input and outputs $A_{i,j}$. 
        \item \textsc{Norm}$(i)$, which takes $i \in [m]$ as input and outputs $\| A_{i,*} \|$. 
        \item \textsc{Norm}$()$, which outputs $\| A \|_\fr$. 
    \end{enumerate}
    Let $\mathcal{T}(A)$ denote the max time the data structure takes to respond to any query.
\end{definition}

For the sake of simplicity, we will assume all input $\SQ$ data structures respond to queries in $\bigO{1}$ time.
There are data structures that can do this in the word RAM model {\cite[Remark~2.15]{chia2019SampdSubLinLowRankFramework}}.
The dynamic data structures that commonly appear in the QML literature~\cite{kerenidis2016QRecSys} respond to queries in $\bigO{\log(mn)}$ time, so using such versions only increases our running time by a logarithmic factor.

The specific choice of data structure does not appear to be important to the quantum-classical separation, though.
The QRAM and data structure assumptions of typical QML algorithms can be replaced with any {\em state preparation assumption}, which is an assumption implying that quantum states corresponding to input data can be prepared in time polylogarithmic in dimension.
However, to the best of our knowledge, all models admitting efficient protocols that take $v$ stored in the standard way as input and output the quantum state $\ket{v}$ also admit corresponding efficient classical sample and query operations that can replace the data structure described above {\cite[Remark~2.15]{chia2019SampdSubLinLowRankFramework}}.
In other words, classical algorithms in the sample and query access model can be run whenever the corresponding QML algorithm can, assuming that the input data is classical.
Consequently, our results on the quantum speedup for low-rank matrix inversion appear robust to changing quantum input models.

\paragraph{Notations.}
For a vector $v \in \C^n$, $\| v \|$ denotes $\ell_2$ norm.
For a matrix $A \in \C^{m\times n}$, $A^\dagger$, $A^+$, $\| A \|$, and $\|A\|_\fr$ denote the conjugate transpose, pseudoinverse, operator norm, and Frobenius norm of $A$, respectively.
We use $A_{i,j}$ to denote the entry of $A$ at the $i$-th row and $j$-th column.
We use $A_{i,*}$ and $A_{*,j}$ to denote the $i$-th row and $j$-th column of $A$.

We use $\E[\cdot]$ and $\Var[\cdot]$ to denote expectation and variance of a random variable.
Abusing notation, for a random vector $v$, we denote $\Var[v] = \E[\|v - \E[v]\|^2]$.
For a differentiable function $f: \R^n \to \R$, $\nabla f$ denotes the gradient of $f$.
We say $f$ is convex if, for all $x, y \in \R^n$ and $t \in [0,1]$, $f(tx + (1-t)y) \leq tf(x) + (1-t)f(y)$.

We use the shorthand $g \lesssim h$ for $g = \bigO{h}$, and respectively $g \gtrsim h$ and $g \eqsim h$ for $g = \Omega(h)$ and $g = \Theta(h)$.
$\tilde{\mathcal{O}}(T)$ denotes $\mathcal{O}(T\cdot \mathrm{polylog}(T))$ and $\mathcal{O}^*$ denotes big $\mathcal{O}$ notation that hides polynomial dependence on $\eps$ and poly-logarithmic dependence on dimension.
Note that the precise exponent on $\log(mn)$ depends on the choice of word size in our RAM models.
Quantum and classical algorithms have different conventions for this choice, hence we use notation that elides this detail for simplicity.

\subsection{Our results}\label{sec:our_result}

In this paper, we focus on solving the following problem.
Prior work in this line of research has studied this problem in the $\lambda = 0$ case~\cite{chia2019SampdSubLinLowRankFramework}. 

\begin{problem}[Regression with regularization] \label{prob:reg}
Given $A \in \C^{m\times n}$, $b \in \C^m$, and a regularization parameter $\lambda \geq 0$, we define the function $f : \C^n \rightarrow \R$ as
\begin{align*}
    f(x) := \frac{1}{2}(\|Ax - b\|^2 + \lambda\|x\|^2).
\end{align*}
Let\footnote{Note that for $\lambda=0$ the function $f$ might not be strictly convex and so the minimizer is not necessarily unique. Nevertheless, on the image of $A^\dagger$, $f$ is strictly convex. Since we will search a solution within the image of $A^\dagger$, we implicitly restrict $f$ to this subspace, thus we can assume without loss of generality that $f$ is strictly convex.\label{foot:ImageRestriction}} $x^* := \arg\min_{x \in \C^n } f(x) = (A^\dagger A + \lambda I)^+ A^\dagger b$.
\end{problem}
\noindent
We will manipulate vectors by manipulating their \emph{sparse descriptions}, defined as follows:
\begin{definition}[sparse description] \label{def:sparse-description}
    We say we have an $s$-sparse description of $x \in \C^d$ if we have an $s$-sparse $v \in \C^n$ such that $x = A^\dagger v$.
    We use the convention that any $s$-sparse description is also a $t$-sparse description for all $t \geq s$.
\end{definition}
\noindent
Our main result is that we can solve regression efficiently, assuming $\SQ(A)$. \Cref{alg:SGDHHL} and the corresponding Theorem~\ref{main-nonsparse} directly improves the previous best result \cite{chia2019SampdSubLinLowRankFramework} due to Chia, Gily\'{e}n, Li, Lin, Tang and Wang by a factor of $\frac{ \| A \|^{20} }{ \sigma^{20} \eps^2 }$.

\begin{algorithm}
    \caption{Quantum-inspired regression via stochastic gradient descent}\label{alg:SGDHHL}
    \begin{algorithmic}[1]
        \STATEx {\bf Input:} $\SQ(A)$, $\Q(b)$, $\eps$, and $\lambda = \bigO{\nrm{A}_F\nrm{A}}$.
        \STATE {\bf Init: } Sparsify $b$ as described in \Cref{mm-approx} 
        to obtain an $s=800\frac{\|A\|_\fr^2\|b\|^2}{(\sigma^2+\lambda)^2\eps^2\|x^*\|^2}$-sparse $\hat{b}$.\label{line:init}\\
        Set $v^{(0)}:=0$ (and implicitly define $x^{(t)}=A^\dagger v^{(t)}$).\\
        Set $\eta := \frac{\eps^2(\sigma^2+\lambda)}{32\|A\|_\fr^2\|A\|^2 + 16\lambda^2}$
        and $T := \frac{\ln(8/\eps^2)}{\eta  (\sigma^2 + \lambda)}=32\ln\left(\frac{\sqrt{8}}{\eps}\right)\frac{2\|A\|_\fr^2\|A\|^2 + \lambda^2}{\eps^2(\sigma^2+\lambda)^2}$.
        \FOR{$t=0,1,\ldots, T-1$}
            \STATE Sample a row index $r$ according to the row norms $\frac{\nrm{A_{r,*}}^2}{\nrm{A}_F^2}$.\label{line:rowSample}
            \STATE Sample $C:=\frac{\|A\|_\fr^2}{\|A\|^2}$ column indices $\{c_i\}_{i \in [C]}$ i.i.d.\ according to $\frac{|A_{r,c}|^2}{\nrm{A_{r,*}}^2}$.\label{line:columnSample}
            \STATE Define $v^{(t+1)}$ as follows:
        \begin{align} \label{eqn:v-updateSimplified}
        v^{(t+1)} := (1-\eta \lambda)v^{(t)} + \eta\hat{b} - \eta \frac{\|A\|_\fr^2}{\|A_{r,*}\|^2}\Big(\frac1C\sum_{j=1}^C\frac{\|A_{r,*}\|^2}{|A_{r,c_j}|^2} A_{r,c_j}(A_{*,c_j})^\dagger v^{(t)}\Big)e_r .
        \end{align}
        \ENDFOR
        \STATE{\bf Output:} $v^{(T)}$ which is $32\frac{2\|A\|_\fr^2\|A\|^2 + \lambda^2}{(\sigma^2 + \lambda)^2\eps^2}\left(\frac{25\|b\|^2}{2\|A\|^2\|x^*\|^2}+\ln\left(\frac{\sqrt{8}}{\eps}\right)\right)$-sparse; set  $x:=A^\dagger v^{(T)}$.
        \STATEx{\bf Queries to $x$:} $x_j = \sum A^\dagger_{j,i}v^{(T)}_i$, so query $A_{i,j}$ for all non-zero $v^{(T)}_i$ and compute the sum.
        \STATEx{\bf Sampling from $|x_j|^2/\nrm{x}^2$:} Perform rejection sampling according to \Cref{lem:matvec}.
    \end{algorithmic}
\end{algorithm}

\begin{theorem}[Main result]\label{main-nonsparse}
Suppose we are given $\SQ(A) \in \C^{m\times n}$, $\Q(b) \in \C^n$.
Denote $\sigma := \|A^+\|^{-1}$ and consider $f(x)$ for $\lambda \lesssim \|A\|_\fr\|A\|$.
\Cref{alg:SGDHHL} runs in
\begin{align*}
    \bigO{\frac{\|A\|_\fr^6\|A\|^2}{(\sigma^2+\lambda)^4\eps^4}\Big( \frac{\|b\|^2}{\|A\|^2\|x^*\|^2} + \log\frac{1}{\eps} \Big) \log\frac{1}{\eps}}
\end{align*}
time and outputs an $\bigO{\frac{\|A\|_\fr^2\|A\|^2}{(\sigma^2+\lambda)^2\eps^2}(\frac{\|b\|^2}{\|A\|^2\|x^*\|^2} + \log\frac{1}{\eps})}$-sparse description of an $x$ such that $\|x - x^*\| \leq \eps\|x^*\|$ with probability $\geq 0.9$.
This description admits $\SQ(x)$ for 
\begin{align*}
    \mathcal{T}(x) = \bigO{\frac{\|A\|_\fr^6\|A\|^6}{(\sigma^2+\lambda)^6\eps^4}\log^2\frac{1}{\eps}\Big(\frac{\|b\|^4}{\|A\|^4\|x^*\|^4} + \log^2\frac{1}{\eps}\Big)\Big(\frac{\|b\|^2}{\|A\|_\fr^2\|x^*\|^2} + 1\Big)}.
\end{align*}
\end{theorem}
Note that the running time stated for $\mathcal{T}(x)$ corresponds to the running time of \textsc{Sample} and \textsc{Norm}; the running time of \textsc{Query} is $\bigO{\frac{\|A\|_\fr^2\|A\|^2}{(\sigma^2+\lambda)^2\eps^2}(\frac{\|b\|^2}{\|A\|^2\|x^*\|^2} + \log\frac{1}{\eps})}$, the sparsity of the description.
So, unlike previous quantum-inspired algorithms, it may be the case that the running time of the $\SQ$ query dominates, though it's conceivable that this running time bound is an artifact of our analysis.

Factors of $\frac{\|b\|^2}{\|A\|^2\|x^*\|^2}$ (which is always at least one) arise because sampling error is additive with respect to $\|b\|$, and need to be rescaled relative to $\|x^*\|$; the quantum algorithm must also pay such a factor.
The expression $\frac{\|A\|^2\|x^*\|^2}{\|b\|^2}$ can be thought of as the fraction of $b$'s ``mass'' that is in the well-conditioned rowspace of $A$.
To make this rigorous, we define the following \emph{thresholded projector}.
\begin{definition} \label{def:threshproj}
    For an $A \in \mathbb{C}^{m\times n}$ with singular value decomposition $A = \sum \sigma_i u_iv_i^\dagger$ and $\lambda \geq 0$, we define $\Pi_{A,\lambda} = \sum p_{A,\lambda}(\sigma_i) u_iu_i^\dagger$ where
    \begin{align*}
        p_{A,\lambda}(\sigma) = \begin{cases}
            0 &  \sigma = 0 \\        
            \frac{2\sigma\sqrt{\lambda}}{\sigma^2+\lambda} & 0 < \sigma \leq \sqrt{\lambda} \\
            1 & \sqrt{\lambda} < \sigma.
        \end{cases}
    \end{align*}
\end{definition}
\noindent For intuition, $\Pi_{A,0}$ projects onto the rowspace of $A$, and for nonzero $\lambda$ smoothly projects away the $u_i$'s with singular values $\sigma_i$ smaller than $\sqrt{\lambda}$ roughly linearly, since $\frac{\sigma}{\sqrt{\lambda}} \leq p_{A,\lambda}(\sigma) \leq \frac{2}{\sqrt{\lambda}}$ for $\sigma \in [0,\sqrt{\lambda}]$.
We use the definition presented here because $\|\Pi_{A,\lambda} b\|$ gives a natural lower bound for $\|x^*\|$: by \cref{eqn:xstarLower}, $\frac{\|A\|^2\|x^*\|^2}{\|b\|^2} \geq \frac{\|\Pi_{A,\lambda}b\|^2}{2\|b\|^2}$ when $\lambda \leq \|A\|^2$.
This is the ``fraction of mass'' type quantity that one would expect.\footnote{Note that this also means $\frac{\|A\|^2\|x^*\|^2}{\|b\|^2} \geq \frac{\|Mb\|^2}{2\|b\|^2}$ for any matrix $M$ satisfying $M \preceq \Pi_{A,\lambda}$. In particular, it holds when $M$ projects onto only those $u_i$'s such that $\sigma_i \geq \sqrt{\lambda}$, i.e., when $M$ is a true projector onto the well-conditioned subspace of the rowspace of $A$.}

As mentioned previously, we use stochastic gradient descent to solve this problem, for $T:=\bigO{\frac{\|A\|_\fr^2\|A\|^2}{(\sigma^2+\lambda)^2\eps^2}\log\frac{1}{\eps}}$ iterations.
Such optimization algorithms are standard for solving regression in this setting~\cite{gs18,cjst20}: the idea is that, instead of explicitly computing the closed form of the minimizer to $f$, which requires computing a matrix pseudoinverse, one can start at $x^{(0)} = \vec{0}$ and iteratively find $x^{(t+1)}$ from $x^{(t)}$ such that the sequence $\{x^{(t)}\}_{t \in \mathbb{N}}$ converges to a local minimum of $f$.
Updating $x^{(t)}$ by nudging it in the direction of $f$'s gradient produces such a sequence, even when we only use decent (stochastic) estimators of the gradient~\cite{bubeck2015ConvexOpt}.
Because $f$ is convex, it only has one minimizer, so $x^{(t)}$ converges to the global minimum.\textsuperscript{\ref{foot:ImageRestriction}}

\paragraph{Challenges.}

Though the idea of SGD for regression is simple and well-understood, we note that \emph{some standard analyses fail in our setting, so some care is needed in analysing SGD correctly}.
\begin{itemize}
\item First, our stochastic gradient (\cref{def-grad}) has variance that depends on the norm of the current iterate, and projection onto the ball of vectors of bounded norm is too costly, so our analyses cannot assume a uniform bound on the second moment of our stochastic gradient.
\item Second, we want a bound for the final iterate $x^{(T)}$, not merely a bound on averaged iterates as is common for SGD analyses, since using an averaging scheme would complicate the analysis in \cref{sec:sq-out}.
\item Third, we want a bound on the output $x$ of our algorithm of the form $\|x - x^*\| \leq \eps\|x^*\|$.
We could have chosen to aim for the typical bound in the optimization literature of $f(x) - f(x^*) \leq \eps(f(x^{(0)}) - f(x^*))$, but our choice is common in the QML literature, and it only makes sense to have $\SQ(x)$ when $x$ is indeed close to $x^*$ up to multiplicative error.
\item Finally, note that acceleration via methods like the accelerated proximal point algorithm (APPA) \cite{fgks15} framework cannot be used in this setting, since we cannot pay the linear time in dimension necessary to recenter gradients.
\end{itemize}
The SGD analysis of Bach and Moulines~\cite{mb11} meets all of our criteria, which inspired the error analysis of this paper, while we use quantum-inspired ideas for the computation of the stochastic gradients (\cref{def-grad}).

The running time of SGD is only good for our choice of gradient when $b$ is sparse, but through a sketching matrix we can reduce to the case where $b$ has $\bigO{\frac{\|A\|_\fr^2\|A\|^2}{\sigma^4\eps^2}\frac{\|b\|^2}{\|AA^+b\|^2}}$ non-zero entries.
Additional analysis is necessary to argue that we have efficient sample and query access to the output, since we need to rule out the possibility that our output $x$ is given as a description $A^\dagger \tilde{x}$ such that $\tilde{x}$ is much larger than $x$, in which case computing, say, $\|x\|$ via rejection sampling would be intractable.

Our analysis of the stochastic gradient descent is completely self-contained and the analysis of the stochastic gradients only relies on a few standard results about quantum-inspired sampling techniques, so the paper shall be accessible even without having a background on convex optimization.

\subsection{Discussion}

\paragraph{Comparisons to prior work.}
First, we describe the prior work on quantum-inspired classical low-rank linear regression algorithms in more detail.
There are three prior papers here: \cite{gilyen2018QInsLowRankHHL} achieves a running time of $\tilde{\mathcal{O}}(\frac{\|A\|_\fr^6k^6\|A\|^{16}}{\sigma^{22}\eps^6})$ (where $k \leq \frac{\|A\|_\fr^2}{\sigma^2}$ is the rank of $A$), \cite{chia2018QInspiredSubLinLowRankLinEqSolver} achieves an incomparable running time of $\mathcal{O}(\frac{\|A\|_\fr^6\|A\|^{22}}{\sigma^{28}\eps^6})$, and \cite{chia2019SampdSubLinLowRankFramework} achieves the same running time as \cite{chia2018QInspiredSubLinLowRankLinEqSolver} for a more general problem (where $\sigma$ can be chosen to be a threshold, and isn't necessarily the minimum singular value of $A$).
For simplicity of exposition, we only compared our algorithm to the latter running time.
However, our algorithm improves on all previous running times: an improvement of $k^6\frac{\|A\|^{10}}{\sigma^{10}\eps^2}$ on the former, and $\frac{\|A\|^{16}}{\sigma^{16}\eps^2}$ on the latter.

\paragraph{Comparison to \cite{gs18}}
In some sense, this work gives a version of the regression algorithm of Gupta and Sidford~\cite{gs18} (which, alongside \cite{cjst20}, is the current state-of-the-art iterative algorithm for regression in many numerically sparse settings) that trades running time for weaker, more quantum-like input assumptions.
Their algorithm takes $A$ and $b$, with only query access to the input's nonzero entries, and explicitly outputs a vector $x$ satisfying $\|x - x^*\| \leq \eps\|x^*\|$ in $\tilde{\mathcal{O}}^*\!\big(\nnz(A) + \frac{\|A\|_\fr}{\sigma}\nnz(A)^{\frac{2}{3}}n^{\frac16}\big)$ time (shown here after naively bounding numerical sparsity by $n$).
They use a stochastic first-order iterative method, that we adapt by mildly weakening the stochastic gradient to be efficient assuming the quantum-inspired data structure (instead of the bespoke data structure they design for their algorithm).

To elaborate, we make a direct comparison.
Suppose we wish to apply our algorithm from \cref{main-nonsparse} to Gupta and Sidford's setting, meaning that we are given $A$ and $b$ as lists of non-zero entries, and wish to output the vector $x$ in full satisfying $\|x - x^*\| \leq \eps\|x^*\|$.
Because of the difference in setting, we need additional $\bigO{\nnz(A)}$ time to set up the data structures to get $\SQ(A)$, and we also need $\mathcal{O}^*(\frac{\|A\|_\fr^2\|A\|^2}{\sigma^4}n)$ time to query for all the entries of $x$, given the $\SQ(x)$ from the output of \cref{alg:SGDHHL}.
In total, this takes $\bigO{\nnz(A)} + \mathcal{O}^*\!\big(\frac{\|A\|_\fr^6\|A\|^2}{\sigma^8} + \frac{\|A\|_\fr^2\|A\|^2}{\sigma^4}n\big)$ time.
This is worse than Gupta and Sidford's time, even after accounting for the fact that, unlike in the setting we consider, we can apply the acceleration framework in~\cite{fgks15}.
Acceleration speeds up our algorithm and reduces our error dependence to poly-logarithmic in $\frac{1}{\eps}$, at the cost of $\nnz(A)$ time per ``phase'' throughout the logarithmically many phases.
The quantum algorithm performs comparably, since one needs to run it $\mathcal{O}^*(n)$ times so that state tomography can convert copies of $\ket{x}$ to an explicit list of $x$'s entries, resulting in a total running time of $\bigO{\nnz(A)} + \tilde{\mathcal{O}}^*\!\big(\frac{\|A\|_\fr}{\sigma}n\big)$.

However, our algorithm performs better than Gupta and Sidford's algorithm in a weaker setting.
QML algorithms often don't take the form specified above, where the goal is to output a full vector and we assume nothing about the input.
For example, the quantum recommendation system~\cite{kerenidis2016QRecSys} only outputs a measurement from the output vector, and works in a setting (providing online recommendations, in essence a dynamic data structure problem) where the $\nnz(A)$ cost for building a data structure can be amortized over many runs of their algorithm.
This is the setting we design our algorithm for, where we are given $\SQ(A)$ directly and only wish to know a measurement from the output vector via $\SQ(x)$, and we achieve a runtime independent of dimension via \cref{alg:SGDHHL}.
Gupta and Sidford's algorithm (even without acceleration) does not behave well in this setting, and would have a running time polynomial in input size, thus being exponentially slower than the quantum-inspired (and quantum) algorithm.
Our goal with this result is to advance the complexity theoretic understanding of quantum speedups by \emph{limiting the possibility of quantum speedup in the most general setting possible}.
To this end, we assume fairly little about our input, so our classical algorithms can work for many different applications where one might hope for quantum speedups.
This comes at the cost of increased runtime, including a polynomial dependence on $\eps^{-1}$ which, being exponentially larger than in other classical algorithms, potentially reduces the applicability of our result as an algorithm.
On the other hand, a similar drawback is also present in the related quantum algorithms, so the resulting runtimes can be at least more directly compared.


Finally, we note that sketching algorithms could also have been used to improve upon previous quantum-inspired results.
In particular, one can view the $\ell_2$-type samples that our input data structure supports as $\frac{\|A\|_\fr^2}{k\sigma^2}$-oversampled leverage score samples, enabling the application of typical sketching results~\cite{w14} associated with approximating matrix products and approximate subspace embeddings.
However, we were not quite able to find an algorithm with running time as good as \cref{main-nonsparse} through the sketching literature, though something similar might be possible (say, the same running time but with the alternate guarantee that $\|Ax - b\| \leq (1+\eps)\|Ax^* - b\|$).
Our guarantee that $\|x - x^*\| \leq \eps\|x^*\|$ is perhaps more natural in this quantum-inspired setting, since we need $x$ to be bounded away from zero to enable efficient (rejection) sampling.
Sketching may be more beneficial in the slightly less restrictive setting where one wishes to find \emph{any} dynamic data structure that can store $A, b$ and output measurements from $\ket{\approx A^+b}$.
In this setting, one wouldn't be required to use the QRAM data structure, and could build alternative data structures that could exploit, for example, oblivious sketches.\footnote{To our knowledge, prior versions of \cite{cchw20} produced algorithms in this setting, but the current version only uses the standard QRAM data structure.}

\paragraph{Comparisons to concurrent and subsequent work.}
Around the same time as this work, Chepurko, Clarkson, Horesh, and Woodruff described algorithms for linear regression and ``quantum-inspired'' recommendation systems, heavily using technical tools from the sketching literature~\cite{cchw20}.
Their algorithm gives a roughly $\frac{\|A\|_\fr^2}{\|A\|^2}$ improvement over our algorithm, ignoring minor details involving problem-specific parameters.\footnote{This comparison is between \cite[Theorem 18]{cchw20}, which, upon taking $\lambda = 0$ and $d' = 1$, gives a runtime of $\bigOt{\eps^{-4}\|A\|_\fr^4\|A\|^4\sigma^{-8}\log(d)}$ to return their output description, and our \cref{main-nonsparse}, which gives a runtime of $\bigOt{\eps^{-4}\|A\|_\fr^6\|A\|^2\sigma^{-8}}$ to return our description of $x$.}
This is done using an algorithm similar to that of \cite{gilyen2018QInsLowRankHHL}, but with an improved analysis and an application of conjugate gradient to solve a system of linear equations (instead of inverting the matrix directly).

Subsequent work by Shao and Montanaro~\cite{sm21} gives an algorithm improving on this result by a factor of $\frac{\|A\|^4}{\sigma^4\eps^2}$ when $\lambda = 0$ and $Ax^* = b$ exactly.
Their work uses the randomized Kaczmarz method~\cite{sv08}, which is a version of stochastic gradient descent~\cite{nsw13}.
If one changes our stochastic gradient $\nabla g$ to the randomized Kaczmarz update used in \cite{sm21}, replacing $A^\dagger b$ with the sketch $\frac{\|A\|_\fr^2}{\|A_{r,*}\|^2}A_{r,*}^\dagger b_r$, (i.e., replacing $\hat{b}$ in \Cref{eqn:v-updateSimplified} by $\frac{\|A\|_\fr^2}{\|A_{r,*}\|^2}b_r e_r$) then our analysis recovers a similar result for the special case when $\lambda = 0$ and $Ax^* = b$.\footnote{Specifically, the bottleneck of our algorithm is the large bound on the residual $\E[\|\nabla g(x^*)\|^2]$, since this is the piece of the error that does not decrease exponentially. If $\lambda=0$ and $Ax^* = b$, then making the described change leads to $Ax$ being replaced by $Ax - b$ in the variance bound of \cref{lem:acc-variance}, which is then zero for $x = x^*$. So, unlike in the original analysis, increasing $C$ beyond $\frac{\|A\|_\fr^2}{\|A\|^2}$ actually helps to reduce this residual. For example, increasing $C$ by a factor of $\frac{1}{\eps^2}$ reduces the number of iterations by a factor of $\frac{1}{\eps^2}$, decreasing the runtime (which is linear in $C$ and quadratic in iteration count) by a factor of $\frac{1}{\eps^2}$. We point the reader to \cite{nsw13} for further exploration of improving the residual in SGD and randomized Kaczmarz.}

\paragraph{Outlook.}
Our result and the aforementioned follow-up work show that simple iterative algorithms like stochastic gradient descent can be nicely combined with quantum-inspired linear algebra techniques, resulting in quantum-inspired algorithms that compare favorably to direct approaches based on finding approximate singular value decompositions~\cite{chia2019SampdSubLinLowRankFramework}.
However, one needs to be careful, as more involved iterative steps require more involved quantum-inspired subroutines that can rapidly blow up the overall complexity, as occurs in the quantum-inspired SDP solver of \cite{chia2019SampdSubLinLowRankFramework}, for example.
We leave it to future work to study which other iterative approaches can be similarly advantageously combined with quantum-inspired linear algebra techniques.
\section{Proofs}

For simplicity we assume knowledge of $\|A\|$, $\sigma$, and later, $\|x^*\|$, exactly, despite that our $\SQ(A)$ only gives us access to $\|A\|_\fr$.
Just an upper bound on $\|A\|$ and lower bounds on $\sigma$ and $\|x^*\|$ respectively suffice, giving a running time bound by replacing these quantities in our complexity bounds by any respective upper or lower bounds.
This holds because these quantities are only used to choose the internal parameters $s, \eta^{-1}, T,$ and $C$, and replacing these quantities by their respective bounds only increase these internal parameters, only improving the resulting error bounds of \Cref{alg:SGDHHL}.

The algorithm we use to solve regression is stochastic gradient descent (SGD).
Suppose we wish to minimize a convex function $f: \R^n \to \R$.
Consider the following recursion, starting from $x^{(0)} \in \R^n$, with a random function $\nabla g: \R^n \to \R$ and a deterministic sequence of positive scalars $(\eta_t )_{k \geq 1}$.
\begin{equation} \label{eqn:sgd}
    x^{(t)} = x^{(t-1)} - \eta_t  \nabla g(x^{(t-1)})
\end{equation}
The idea behind SGD is that if $\nabla g$ is a good stochastic approximation to $\nabla f$, then $x^{(t)}$ will converge to the minimizer of $f$.

In our case, as defined in the introduction, $f(x) = \frac12(\|Ax - b\|^2 + \lambda\|x\|^2)$, so we need a stochastic approximation to $\nabla f(x) = A^\dagger Ax - A^\dagger b + \lambda x$.
Our choice of stochastic gradient comes from observing that $A^\dagger A x = \sum_{r=1}^m\sum_{c=1}^n (A_{r,*})^\dagger A_{r,c} x_c$, so we can estimate this by sampling some of the summands.

\begin{restatable}{definition}{grad} \label{def-grad}
We define $\nabla g(x)$ to be the random function resulting from the following process.
Draw $r\in [m]$ from the distribution that is $r$ with probability $\frac{\|A_{r,*}\|^2}{\|A\|_\fr^2}$, and then draw $c_1,\ldots,c_C$ i.i.d.\ from the distribution that is $c$ with probability $\frac{|A_{r,c}|^2}{\|A_{r,*}\|^2}$.
For this choice of $r$ and $c_1,\ldots,c_C$, take
\begin{equation}
    \nabla g(x) = \frac{\|A\|_\fr^2}{\|A_{r,*}\|^2}\Big( \frac{1}{C} \sum_{j=1}^C\frac{\|A_{r,*}\|^2}{ |A_{r,c_j}|^2} A_{r,c_j}x_{c_j}\Big) (A_{r,*})^\dagger - A^\dagger b + \lambda x.
\end{equation}
\end{restatable}
\noindent Notice that the first term of this expression is the average of copies of the random vector
\begin{align*}
    \frac{\|A\|_\fr^2}{|A_{r,c}|^2} A_{r,c}x_{c}(A_{r,*})^\dagger \quad \text{ with probability } \frac{|A_{r,c}|^2}{\|A\|_\fr^2}.
\end{align*}
\Cref{alg:SGDHHL} is simply running the iteration \cref{eqn:sgd} with $x^{(0)} = \vec{0}$, $\nabla g$ as defined in \cref{def-grad}, and $\eta_t :=  \frac{\eps^2(\sigma^2+\lambda)}{32\|A\|_\fr^2\|A\|^2 + 16\lambda^2}$ (as we will see, this comes from applying \cref{thm:mainFixed} with $\eps'=\eps/2$).
In \cref{sec:sgd}, we prove that $x^{(T)}$ satisfies the desired error bound for $T := \Theta\big(\frac{\|A\|_\fr^2\|A\|^2+\lambda^2}{(\sigma^2+\lambda)^2\eps^2}\log\frac{1}{\eps}\big)$.
In \cref{sec:sq-out}, we prove that computing $x^{(T)}$ and simulating $\SQ(x^{(T)})$ can be done in the desired running time, assuming that $b$ is sparse.
Finally, in \cref{sec:sketch}, we show how to generalize our result to non-sparse $b$.

\subsection{Error analysis of stochastic gradient descent} \label{sec:sgd}

We now list the properties of the stochastic gradient from \cref{def-grad}: the proofs are straightforward computation and so are deferred to the appendix.
\begin{restatable}{lemma}{accvariance} \label{lem:acc-variance}
For fixed $x,y \in \C^n$ and the random function $\nabla g(\cdot)$ defined in \cref{def-grad}, the following properties hold ($\nabla g(x)$ and $\nabla g(y)$ use the same instance of the random function):
\begin{align*}
   {\bf Part~1.}&&\E [\nabla g(x)] &= \nabla f(x) = A^\dagger A x - A^\dagger b + \lambda x \\
   {\bf Part~2.}&& \Var [\nabla g(x)] &  = \frac{1}{C} \|A\|_\fr^4 \|x\|^2 + \Big(1-\frac{1}{C}\Big)\|A\|_\fr^2\|Ax\|^2 - \norm{A^\dagger Ax}^2 \\
   {\bf Part~3.}&& \E[\norm{\nabla g(x) - \nabla g(y)}_2^2] &= \norm{(A^\dagger A + \lambda I)(x-y)}^2 + \Var[ \nabla g(x-y) ]
\end{align*}
\end{restatable}

We take $C = \|A\|_\fr^2/\|A\|^2$, which is the largest value we can set $C$ to before it stops having an effect on the variance of $\nabla g$.
In \cref{sec:sq-out}, we will see that it's good for SGD's running time to take $C$ to be as large as possible.
In this setting, and more generally if $C \leq \frac{\|A\|_\fr^2\|x\|^2}{\|Ax\|^2}$ the variance in Part 2 can be bounded as $\frac{2}{C}\|A\|_\fr^4\|x\|^2$.

Now, we show that performing SGD for $T$ iterations gives an $x$ sufficiently close to the optimal vector $x^*$.

\begin{proposition}\label{thm:mainFixed}
	Consider a matrix $A \in \C^{m\times n}$, a vector $b \in \C^n$, a regularization parameter $\lambda \geq 0$, and an error parameter $\eps \in (0,1]$.
	Denote $\sigma := \|A^+\|^{-1}$.
	Let $x^{(T)}$ be defined as \cref{eqn:sgd}, with $x^{(0)} = \vec{0}$, $\eta_t := \eta := \frac{\eps^2(\sigma^2+\lambda)}{8\|A\|_\fr^2\|A\|^2 + 4\lambda^2}$, and $\nabla g$ as defined in \cref{def-grad} with $C := \frac{\|A\|_\fr^2}{\|A\|^2}$.
	Then for $T := \frac{\ln(2/\eps^2)}{\eta  (\sigma^2 + \lambda)}=8\ln\left(\frac{\sqrt{2}}{\eps}\right)\frac{2\|A\|_\fr^2\|A\|^2 + \lambda^2}{\eps^2(\sigma^2+\lambda)^2}$,
	\begin{align*}
	\E[\|x^{(T)} - x^*\|^2] \leq \eps^2\|x^*\|^2.
	\end{align*}
	In particular, by Chebyshev's inequality, with probability $\geq 0.96$, $\|x^{(T)} - x^*\| \leq 5\eps\|x^*\|$.
\end{proposition}

\begin{proof}
	Stochastic gradient descent is known to require a number of iterations linear in (something like the) second moment of the stochastic gradient.
	To analyze SGD, we loosely follow a strategy used by Bach and Moulines~\cite{mb11}.
	First, note that
	\begin{align*}
	&\E[\|\nabla g(x)\|^2]  \\
	&\leq \E[2\|\nabla g(x) - \nabla g(x^*)\|^2 + 2\|\nabla g(x^*)\|^2]
	\tag*{since $\|a + b\|^2 \leq 2\|a\|^2 + 2\|b\|^2$}\\
	&= 2\Big(\|(A^\dagger A + \lambda I)(x - x^*)\|^2 + \Var[\nabla g(x-x^*)] + \Var[\nabla g(x^*)] \Big)
	\tag*{since $\nabla f(x^*)  = \vec{0}$}\\
	&\leq 2\Big(\|(A^\dagger A + \lambda I)(x - x^*)\|^2 + 2\|A\|_\fr^2\|A\|^2(\|x-x^*\|^2 + \|x^*\|^2) \Big)
	\tag*{by \cref{lem:acc-variance}}\\
	&\leq 2(\|A^\dagger A + \lambda I\|^2 + 2\|A\|_\fr^2\|A\|^2)\|x - x^*\|^2 + 4\|A\|_\fr^2\|A\|^2\|x^*\|^2.
	\end{align*}
	Next, we use the well-known fact that $f(x)$ is $(\sigma^2 + \lambda)$-strongly convex,\footnote{Since we search for a sparse description of a solution vector we are effectively restricting the function to the image of $A^\dagger$, and on this subspace the function is indeed $(\sigma^2 + \lambda)$-strongly convex.} implying that $\langle \nabla f(x) - \nabla f(y), x - y\rangle \geq (\sigma^2 + \lambda)\|x-y\|^2$~\cite[Lemma~3.11]{bubeck2015ConvexOpt}, which can also be directly verified using the formula $\nabla f(x) = A^\dagger Ax - A^\dagger b + \lambda x$ for the gradient.
	\begin{align*}
	&\E[\|x^{(t)} - x^*\|^2 \mid x^{(t-1)}] \\
	&= \E[\|x^{(t-1)} - x^*\|^2 + 2\langle x^{(t-1)} - x^*, x^{(t)} - x^{(t-1)}\rangle + \|x^{(t)} - x^{(t-1)}\|^2 \mid x^{(t-1)}] \\
	&= \E[\|x^{(t-1)} - x^*\|^2 - 2\eta_t\langle x^{(t-1)} - x^*, \nabla g(x^{(t-1)})\rangle + \eta_t^2\|\nabla g(x^{(t-1)})\|^2 \mid x^{(t-1)}] \\
	&= \|x^{(t-1)} - x^*\|^2 - 2\eta_t\langle x^{(t-1)} - x^*, \nabla f(x^{(t-1)})\rangle + \eta_t^2\E[\|\nabla g(x^{(t-1)})\|^2 \mid x^{(t-1)}] \\
	&\leq (1 - 2\eta_t(\sigma^2 + \lambda) + 2\eta_t^2(\|A^\dagger A + \lambda I\|^2 + 2\|A\|_\fr^2\|A\|^2))\|x^{(t-1)} - x^*\|^2 + 4\eta_t^2\|A\|_\fr^2\|A\|^2\|x^*\|^2.
	\end{align*}
	We conditioned on $x^{(t-1)}$ in the above computation, so the randomness in $\nabla g$ came only from one iteration.
	Let $\delta_t:=\mathbb{E}[\nrm{x^{(t)}-x^*}^2]$.
    Then by taking expectation over $x^{(t-1)}$ for both sides of the above computation, we get
	\begin{align*}
	\delta_t &\leq (1 - 2\eta_t(\sigma^2 + \lambda) + 2\eta_t^2(\|A^\dagger A + \lambda I\|^2 + 2\|A\|_\fr^2\|A\|^2))\delta_{t-1} + 4\eta_t^2\|A\|_\fr^2\|A\|^2\|x^*\|^2 \\
	&\leq (1 - 2\eta_t(\sigma^2 + \lambda) + 2\eta_t^2(4\|A\|_\fr^2\|A\|^2 + 2\lambda^2))\delta_{t-1} + 4\eta_t^2\|A\|_\fr^2\|A\|^2\|x^*\|^2,
	\end{align*}
    where the last step follows from the inequality $\|a + b\|^2 \leq 2\|a\|^2 + 2\|b\|^2$.
	Now, by substituting $\eta_t^2 = \eta^2 = \eta\frac{\eps^2(\sigma^2+\lambda)}{8\|A\|_\fr^2\|A\|^2+4\lambda^2}$ and using that $\eps\leq 1$ we get 
	\begin{align*}
		\delta_t
		&\leq (1 - \eta(\sigma^2 + \lambda))\delta_{t-1} + \frac{\eps^2}{2} \eta (\sigma^2 + \lambda)\|x^*\|^2.
	\end{align*}	
	This is a purely deterministic recursion on $\delta_t$, which we can bound by
	\begin{align*}
	\delta_t
	&\leq \left(\exp(-t\eta (\sigma^2 + \lambda))+\frac{\eps^2}{2}\right)\|x^*\|^2.
	\end{align*}
	Since $x^{(0)} = \vec{0}$ we have $\delta_0 = \|x^*\|^2$ so the bound holds for $t=0$, and by induction we have
	\begin{align}
	\delta_t
	&\leq (1 - \eta(\sigma^2 + \lambda))\left(\exp(-(t-1)\eta (\sigma^2 + \lambda))+\frac{\eps^2}{2}\right)\|x^*\|^2 + \frac{\eps^2}{2} \eta (\sigma^2 + \lambda)\|x^*\|^2 \nonumber\\
	&=\left((1 - \eta(\sigma^2 + \lambda))\exp(-(t-1)\eta (\sigma^2 + \lambda))+\frac{\eps^2}{2}\right)\|x^*\|^2 \nonumber\\
	&\leq\left(\exp(-t\eta (\sigma^2 + \lambda))+\frac{\eps^2}{2}\right)\|x^*\|^2, \label{eqn:internal-iterateFixed}
	\end{align}
	which is at most $\eps^2\|x^*\|^2$ for $t\geq T = \frac{\ln(2/\eps^2)}{\eta(\sigma^2+\lambda)}$.
\end{proof}
\subsection{Time complexity analysis of SGD for sparse \texorpdfstring{$b$}{b}} \label{sec:sq-out}

Now, we show that, assuming $\SQ(A)$, it's possible to perform the gradient steps when $b$ is sparse (that is, the number of non-zero entries of $b$, $\|b\|_0$, is small).
Recall from \cref{def:sparse-description} that we say we have an $s$-sparse description of $x \in \C^d$ if we have an $s$-sparse $v \in \C^n$ such that $x = A^\dagger v$.
\begin{lemma}\label{lem:one_step}
    Given $\SQ( A )$, we can output $x^{(t)}$ as a $(t + \|b\|_0)$-sparse description in time
    \begin{align*}
     \bigO{ C t ( t + \|b\|_0 ) }.
    \end{align*}
\end{lemma}
\begin{proof}
First, suppose we are given $x^{(t)}$ as an $s$-sparse description, and wish to output a sparse description for the next iterate $x^{(t+1)}$, which from \cref{def-grad}, satisfies
\begin{align*}
    x^{(t+1)} &= x^{(t)} - \eta_{t+1} \nabla g(x^{(t)}) \\
    &= x^{(t)} - \eta_{t+1} \cdot \Big(\frac{\|A\|_\fr^2}{\|A_{r,*}\|^2}\Big(\frac1C\sum_{j=1}^C\frac{\|A_{r,*}\|^2}{|A_{r,c_j}|^2} A_{r,c_j}x_{c_j}^{(t)}\Big) (A_{r,*})^\dagger - A^\dagger b + \lambda x^{(t)}\Big).
\end{align*}
The $r$ and $c_j$'s are drawn from distributions, and $\SQ(A)$ can produce such samples with its \textsc{Sample} queries, taking $\bigO{C}$ time.

From inspection of the above equation, if we have $x^{(t)}$ in terms of its description as $A^\dagger v^{(t)}$, then we can write $x^{(t+1)}$ as a description $A^\dagger v^{(t+1)}$ where $v^{(t+1)}$ satisfies (as in \Cref{eqn:v-updateSimplified})
\begin{align} \label{eqn:v-update}
    v^{(t+1)} = v^{(t)} - \eta_{t+1} \cdot \Big(\frac{\|A\|_\fr^2}{\|A_{r,*}\|^2}\Big(\frac1C\sum_{j=1}^C\frac{\|A_{r,*}\|^2}{|A_{r,c_j}|^2} A_{r,c_j}(A_{*,c_j})^\dagger v^{(t)}\Big)e_r - b + \lambda v^{(t)}\Big).
\end{align}

Here, $e_r$ is the vector that is one in the $r$th entry and zero otherwise.
So, if $v^{(t)}$ is $s$-sparse, and has a support that includes the support of $b$, then $v^{(t+1)}$ is $(s+1)$-sparse.
Furthermore, by exploiting the sparsity of $v^{(t)}$, computing $v^{(t+1)}$ takes $\bigO{Cs}$ time (including the time taken to use $\SQ(A)$ to query $A$ for all of the relevant norms and entries).

So, if we wish to compute $x^{(t)}$, we begin with $x^{(0)}$, which we have trivially as an $\|b\|_0$-sparse description ($v^{(0)} = \vec{0}$).
It is sparser, but if we consider $x^{(0)}$ as having the same support as $b$, by the argument described above, we can then compute $x^{(1)}$ as an $(\|b\|_0+1)$-sparse description in $\bigO{C\|b\|_0}$ time.

By iteratively computing $v^{(i+1)}$ from $v^{(i)}$ in $\bigO{C(\|b\|_0 + i)}$ time, we can output $x^{(t)}$ as a $(t + \|b\|_0)$-sparse description in
\begin{align*}
    \mathcal{O}\Big(\sum_{i=0}^{t} C(\|b\|_0 + i)\Big) = \bigO{Ct (t+\|b\|_0)}
\end{align*}
time as desired.
\end{proof}

If $\lambda = O(\|A\|_\fr\|A\|)$ by \cref{thm:mainFixed}, to get a good enough iterate $x^{(T)}$ we can take $T = \bigO{\frac{\|A\|^2\|A\|_\fr^2}{\left(\sigma^4+\lambda^2\right)\eps^2}\log\frac{1}{\eps}}$ and $C = \frac{\|A\|_\fr^2}{\|A\|^2}$, giving a running time of
\begin{align*}
    \bigO{\|b\|_0\frac{\|A\|_\fr^4}{\left(\sigma^4+\lambda^2\right)\eps^2}\log\frac{1}{\eps} + \frac{\|A\|^2\|A\|_\fr^6}{\left(\sigma^8+\lambda^4\right)}\log^2\frac{1}{\eps}}.
\end{align*}
Notice that it's good to scale $C$ up to be as large as possible, since it means a corresponding linear decrease in the number of iterations, which our algorithm's running time depends on quadratically.

After performing SGD, we have an $x^{(T)}$ that is close to $x^*$ as desired, and it is given as a sparse description $A^\dagger v^{(T)}$.
We want to say that we have $\SQ(x^{(T)})$ from its description.
For this we invoke a result from \cite{tang2018QuantumInspiredRecommSys}, describing how to length-square sample from a vector that is a linear combination of length-square accessible vectors---that is, getting $\SQ(c_1v_1 + \cdots + c_dv_d)$ given $\SQ(v_1),\ldots,\SQ(v_d)$.

\begin{lemma}[{\cite[Lemmas 2.9 and 2.10]{chia2019SampdSubLinLowRankFramework}}] \label{lem:matvec}
Suppose we have $\SQ(M^\dagger)$ for $M \in \C^{n \times d}$ and $\Q(x)\in\C^d$.
Denote $y := Mx$ and $\Delta := \sum_{i=1}^d \|M_{*,i}\|^2x_i^2 / \| y \|_2^2$.
Then we can implement $\SQ(y) \in \C^n$ with $\mathcal{T}(y) = \bigO{d^2\Delta\log( 1/\delta) \cdot \mathcal{T}(M)}$, where queries succeed with probability $\geq 1-\delta$.
Namely, we can:
\begin{enumerate}[label=(\alph*)]
    \item\label{it:query} query for entries with complexity $\bigO{d\cdot\mathcal{T}(M)}$;
    \item\label{it:sample} sample from $y$ with running time $\tsample(y)$ satisfying
    \begin{gather*}
    \E[\tsample(y)] = \bigO{d^2\Delta\cdot\mathcal{T}(M)} \\
    \text{and }\Pr[ \tsample(y) = \bigO{d^2\Delta \cdot \log ( 1 / \delta )} ] \geq 1 - \delta.
    \end{gather*}
    \item\label{it:norm} estimate $\nrm{y}$ to $(1\pm \eps)$ multiplicative error with success probability at least $1-\delta$ in complexity 
    \begin{align*}
    \bigO{\frac{d^2\Delta}{\eps^2}\mathcal{T}(M) \cdot \log(1/\delta)}.
    \end{align*}
\end{enumerate}
\end{lemma}

So we care about the quantity $\Delta = \frac{1}{\|A^\dagger v^{(t)}\|^2}\sum_{i=1}^m \|A_{i,*}\|^2\abs{v^{(t)}_i}^2$, with $t = T$, where $v^{(t)}$ follows the recurrence according to \cref{eqn:v-update} (recalling from before that $r$ and $c_1,\ldots,c_C$ are sampled randomly and independently each iteration):
\begin{align*}
    v^{(t+1)} = v^{(t)} - \eta_{t+1} \cdot \Big(\underbrace{\frac{\|A\|_\fr^2}{\|A_{r,*}\|^2}\Big(\frac1C\sum_{j=1}^C\frac{\|A_{r,*}\|^2}{|A_{r,c_j}|^2} A_{r,c_j}(A_{*,c_j})^\dagger v^{(t)}\Big)e_r - b + \lambda v^{(t)}}_{\nabla \tilde{g}(v^{(t)})}\Big).
\end{align*}
Roughly speaking, $\Delta$ encodes the amount of cancellation that could occur in the product $A^\dagger v^{(t)}$.
We will consider it as a norm: $\Delta = \frac{\|v^{(t)}\|_D^2}{\|x^{(t)}\|^2}$, for $D$ a diagonal matrix with $D_{ii} = \|A_{i,*}\|$, where $\|v\|_D := \sqrt{v^\dagger D^\dagger D v}$.
The rest of this section will be devoted to bounding $\Delta$ where $x^{(t)}$ comes from SGD as in \cref{thm:mainFixed}.

First, notice that we can show similar moment bounds for $\nabla \tilde{g}(v^{(t)})$ in the $D$ norm as those for $\nabla g(x^{(t)})$ in \cref{lem:acc-variance}; we defer the proof of this to the appendix.
\begin{restatable}{lemma}{variancevariant}
\label{lem:variance-variant}
Let $\nabla \tilde{g}(v)$ denote
\begin{align*}
    \frac{\|A\|_\fr^2}{\|A_{r,*}\|^2}\Big(\frac1C\sum_{j=1}^C\frac{\|A_{r,*}\|^2}{|A_{r,c_j}|^2} A_{r,c_j}(A_{*,c_j})^\dagger v\Big)e_r - b + \lambda v,
\end{align*}
where $r, c_1,\ldots,c_C$ are sampled according to the same distribution as $\nabla g$ is (\cref{def-grad}).
For a fixed $v \in \C^m$ and $C = \frac{\|A\|_\fr^2}{\|A\|^2}$, we have:
\begin{align*}
{\bf Part~1.} &&\E[\nabla \tilde{g}(v)] &= AA^\dagger v - b + \lambda v \\
{\bf Part~2.} &&\E[\|\nabla \tilde{g}(v) - \E[\nabla \tilde{g}(v)]\|_D^2] &\leq 2\|A\|_\fr^2\|A\|^2\|A^\dagger v\|^2.
\end{align*}
\end{restatable}
\noindent
To bound $\Delta$, we will show a recurrence via \cref{lem:variance-variant}.
Note that 
\begin{align*}
    &\E[\|v^{(t+1)}\|_D\mid v^{(t)}] \\
    &\leq \|v^{(t)}\|_D(1-\eta_{t+1}\lambda) + \eta_{t+1} \E[\|\nabla \tilde{g}(v^{(t)})-\lambda v^{(t)}\|_D \mid v^{(t)}]
    \tag*{by triangle inequality} \\
    &\leq \|v^{(t)}\|_D + \eta_{t+1} \sqrt{\E[\|\nabla \tilde{g}(v^{(t)})-\lambda v^{(t)}\|_D^2 \mid v^{(t)}]}
    \tag*{by $\E[Z]^2 \leq \E[Z^2]$} \\
    &\leq \|v^{(t)}\|_D + \eta_{t+1} \sqrt{\|AA^\dagger v^{(t)} - b\|_D^2 + 2\|A\|_\fr^2\|A\|^2\|A^\dagger v^{(t)}\|^2}
    \tag*{by $\E[Z^2] = \E[Z]^2 + \Var[Z]$} \\
    &\leq \|v^{(t)}\|_D + \eta_{t+1}(\|AA^\dagger v^{(t)} - b\|_D + \sqrt{2}\|A\|_\fr\|A\|\|A^\dagger v^{(t)}\|)
    \tag*{by $\sqrt{a+b} \leq \sqrt{a} + \sqrt{b}$}\\
    &\leq \|v^{(t)}\|_D + \eta_{t+1}(\|b\|_D + (1+\sqrt{2})\|A\|_\fr\|A\|\|x^{(t)}\|). \tag*{by $x^{(t)} = A^\dagger v^{(t)}$ and $\|u\|_D \leq \|A\|\|u\|$}
\end{align*}
Taking expectation over $v^{(t)}$ we get
\begin{align*}
    \E[\|v^{(t+1)}\|_D] &\leq \E[\|v^{(t)}\|_D] + \eta_{t+1}(\|b\|_D + 3\|A\|_\fr\|A\|\E[\|x^{(t)}\|]),
\end{align*}
and since $v^{(0)} = \vec{0}$ we can trivially solve this recurrence resulting in
\begin{align} \label{eqn:v-recurrence-soln}
    \E[\|v^{(t)}\|_D] \leq \sum_{t'=1}^t \eta_{t'}(\|b\|_D + 3\|A\|_\fr\|A\|\E[\|x^{(t'-1)}\|]).
\end{align}
Further, using the parameter choices of \Cref{thm:mainFixed} due to \cref{eqn:internal-iterateFixed} we have
\begin{align*}
    \E[\|x^{(t')}\|] \leq \|x^*\| + \sqrt{\E[\|x^{(t')} - x^*\|^2]}
    = \|x^*\| + \sqrt{\delta_{t'}}
    \leq 3\|x^*\|,
\end{align*}
so
\begin{align} \label{eqn:cancel-final}
    \E[\|v^{(t)}\|_D] 
    &\leq t\eta (\|b\|_D + 9\|A\|_\fr\|A\|\|x^*\|).
\end{align}

Finally, we can bound our cancellation constant.
Using Markov's inequality for \cref{eqn:cancel-final} and combining it with \cref{thm:mainFixed} by a union bound we get that with probability $\geq 0.9$, both $\|x^{(T)} - x^*\| \leq 5\eps\|x^*\|$ and
\begin{align*}
    \|v^{(T)}\|_D \leq 20T\eta(\|b\|_D + 9\|A\|_\fr\|A\|\|x^*\|)
    = 20\frac{\ln(2/\eps^2)}{(\sigma^2 + \lambda)}(\|b\|_D + 9\|A\|_\fr\|A\|\|x^*\|)
\end{align*}
When both those bounds hold and $\eps \leq 1/10$, we have
\begin{multline}\label{eqn:del-boundFixed}
    \Delta = ~ \frac{\|v^{(T)}\|_D^2}{\|x^{(T)}\|^2}
    \leq \frac{\|v^{(T)}\|_D^2}{\|x^*\|^2(1-5\eps)^2}  
    \lesssim ~ \frac{\log^2(1/\eps)}{(\sigma^2 + \lambda)^2}\left(\frac{\|b\|_D^2}{\|x^*\|^2} + \|A\|^2_\fr\|A\|^2\right) \\
    \lesssim ~ \frac{\log^2(1/\eps)}{(\sigma^2 + \lambda)^2}\left(\frac{\|A\|^2\|b\|^2}{\|x^*\|^2} + \|A\|^2_\fr\|A\|^2\right). 
\end{multline}
The last inequality follows from using that $\|b\|_D^2 = \|Db\|^2 \leq \|D\|^2\|b\|^2 \leq \|A\|^2\|b\|^2$.

\subsection{Extending to non-sparse \texorpdfstring{$b$: Proof of \cref{main-nonsparse}}{b: Proof of Theorem 1.4}} \label{sec:sketch}
In previous sections, we have shown how to solve our regularized regression problem for sparse $b$: from \cref{thm:mainFixed}, performing SGD for $T \eqsim \frac{\|A\|_\fr^2\|A\|^2+\lambda^2}{(\sigma^2+\lambda)^2\eps^2}\log\frac{1}{\eps}$ iterations outputs an $x$ with the desired error bound; from \cref{lem:one_step}, it takes $\bigO{\frac{\|A\|_\fr^2}{\|A\|^2}T(T+\|b\|_0)}$ time to output $x$ as a sparse description; and from \cref{sec:sq-out}, we have sample and query access to that output $x$ given its sparse description. Let $\lambda = \bigO{\nrm{A}_F\nrm{A}}$ in the rest of the section.

Now, all that remains is to extend this work to the case that $b$ is non-sparse.
In this case, we will simply replace $b$ with a sparse $\hat{b}$ that behaves similarly, and show that running SGD with this value of $\hat{b}$ gives all the same results.
The sparsity of $\hat{b}$ will be $\bigO{\frac{\|A\|_\fr^2\|b\|^2}{(\sigma^2 + \lambda)^2\eps^2\|x^*\|^2}}$, giving a total running time of
\begin{align*}
    \bigO{\frac{\|A\|_\fr^6\|A\|^2}{(\sigma^2+\lambda)^4\eps^4}\Big( \frac{\|b\|^2}{\|A\|^2\|x^*\|^2} + \log\frac{1}{\eps} \Big) \log\frac{1}{\eps}}.
\end{align*}
We show below that the bound in \cref{eqn:del-boundFixed} also holds in this case.
Using \cref{lem:matvec}, the time it takes to respond to a query to $\SQ(x^{(T)})$ with probability $0.99$ is $(T+\|\hat{b}\|_0)^2\Delta$, which gives the running time in \cref{main-nonsparse},
\begin{align*}
    \bigO{(T + \|\hat{b}\|_0)^2\Delta}
    = \bigO{\frac{\|A\|_\fr^6\|A\|^6}{(\sigma^2+\lambda)^6\eps^4}\log^2\frac{1}{\eps}\Big(\frac{\|b\|^4}{\|A\|^4\|x^*\|^4} + \log^2\frac{1}{\eps}\Big)\Big(\frac{\|b\|^2}{\|A\|_\fr^2\|x^*\|^2} + 1\Big)}.
\end{align*}
The crucial observation for sparsifying $b$ is that we can use importance sampling to approximate the matrix product $A^\dagger b$, which suffices to approximate the solution $x^*$.
\begin{lemma}[Matrix multiplication to Frobenius norm error, {\cite[Lemma~4]{dkm06}}] \label{mm-approx}
    Consider $X \in \C^{m\times n}, Y \in \C^{m\times p}$, and let $S \in \R^{s\times m}$ be an \emph{importance sampling matrix} for $X$.
    That is, let each $S_{i,*}$ be independently sampled to be $\frac{\|X\|_\fr}{\sqrt{s}\|X_{i,*}\|}e_i$ with probability $\frac{\|X_{i,*}\|^2}{\|X\|_\fr^2}$.
    Then
    \begin{align*}
    \E[\|X^\dagger S^\dagger SY - X^\dagger Y\|_\fr^2] \leq & ~ \frac{1}{ s} \|X\|_\fr^2\|Y\|_\fr^2. 
    \end{align*}
\end{lemma}

We have $\SQ(A)$, so we can use \cref{mm-approx} with $s \leftarrow 200\frac{\|A\|_\fr^2\|b\|^2}{(\sigma^2+\lambda)^2\eps^2\|x^*\|^2}$, $X \leftarrow A$, and $Y \leftarrow b$, to find an $S$ in $\bigO{s}$ time that satisfies the guarantee
\begin{align}
     \|A^\dagger S^\dagger S b - A^\dagger b\| &\leq \sqrt{\frac{200}{s}} \|A\|_\fr\|b\| = \eps(\sigma^2+\lambda)\|x^*\| \label{eqn:mm-approx-1}
     \intertext{with probability $\geq 0.995$, using Markov's inequality.
     Recall that $D$ is defined to be the diagonal matrix with $D_{i,i} = \|A_{i,*}\|$, so $D^\dagger = D$, $\|D\|_\fr = \|A\|_\fr$, and importance sampling matrices for $A$ are also importance sampling matrices for $D$.
     Consequently, we can apply \cref{mm-approx} with the same argument to conclude that, with probability $\geq 0.995$,}
     \|D S^\dagger S b - D b\| &\leq \sqrt{\frac{200}{s}} \|D\|_\fr\|b\| = \eps(\sigma^2+\lambda)\|x^*\| \label{eqn:mm-approx-2}
\end{align}
By union bound, both \cref{eqn:mm-approx-1} and \cref{eqn:mm-approx-2} hold \emph{for the same $S$} with probability $\geq 0.99$.
Assuming \cref{eqn:mm-approx-1,eqn:mm-approx-2} hold, we can perform SGD as in \cref{thm:mainFixed} on $\hat{b} := S^\dagger Sb$ (which is $s$-sparse) to find an $x$ such that $\|x - \hat{x}^*\| \leq \eps\|\hat{x}^*\|$, where $\hat{x}^*$ is the optimum $(A^\dagger A + \lambda I)^{-1}A^\dagger\hat{b}$.
This implies that $\|x - x^*\| \leq 2\eps\|x^*\|$, since by \cref{eqn:mm-approx-1},\footnote{We obtain the constants in \Cref{alg:SGDHHL} by using the error parameter $\hat{\eps}:=\eps/2$ to satisfy $\|x - x^*\| \leq 2\hat{\eps}\|x^*\|=\eps\|x^*\|$.}
\begin{align} \label{eqn:sparsify-solutions-close}
    \|\hat{x}^* - x^*\| = \|(A^\dagger A+ \lambda I)^+ A^\dagger(\hat{b} - b)\|
    \leq \frac{1}{\sigma^2+\lambda}\|A^\dagger(\hat{b} - b)\| \leq \eps\|x^*\|.
\end{align}
To bound the running times of $\SQ(x)$, we modify the analysis of $\Delta$ from \cref{sec:sq-out}: recalling from \cref{eqn:del-boundFixed}, we have that, with probability $\geq 0.9$,
\begin{align} \label{eqn:cancel-with-approx}
    \Delta &\lesssim
    \frac{\log^2(1/\eps)}{(\sigma^2 + \lambda)^2}\left(\frac{\|\hat{b}\|_D^2}{\|\hat{x}^*\|^2} + \|A\|^2_\fr\|A\|^2\right)
    \lesssim \frac{\log^2(1/\eps)}{(\sigma^2 + \lambda)^2}\left(\frac{\|b\|_D^2}{\|x^*\|^2} + \|A\|^2_\fr\|A\|^2\right).
\end{align}
In other words, the same bound on $\Delta$ from \cref{eqn:del-boundFixed} holds after approximating $b$ by $\hat{b}$.
The last step follows from the following upper bound on $\frac{\|\hat{b}\|_D^2}{\|\hat{x}^*\|^2}$:
\begin{multline*}
    \frac{\|\hat{b}\|_D^2}{\|\hat{x}^*\|^2}
    \leq \frac{\|\hat{b}\|_D^2}{\|x^*\|^2(1-\eps)^2}
    = \frac{\|DS^\dagger Sb\|^2}{\|x^*\|^2(1-\eps)^2}
    \leq \frac{(\|Db\| + \|Db - DS^\dagger Sb\|)^2}{\|x^*\|^2(1-\eps)^2} \\
    \lesssim \frac{\|Db\|^2 + \eps^2(\sigma^2+\lambda)^2\|x^*\|^2}{\|x^*\|^2(1-\eps)^2}
    \lesssim \frac{\|b\|_D^2}{\|x^*\|^2} + \eps^2(\sigma^2+\lambda)^2
\end{multline*}
where the first inequality follows from \cref{eqn:sparsify-solutions-close}, the second inequality is triangle inequality, the third inequality follows from \cref{eqn:mm-approx-2}, and the last inequality holds for $\eps\leq \frac{1}{2}$ which we can assume without loss of generality.
The added term of $\eps^2(\sigma^2+\lambda)^2$ is dominated by the $\|A\|_\fr^2\|A\|^2$ term in \cref{eqn:cancel-with-approx}.

Finally, we observe the following fact about $\|x^*\|$ which can be used to convert the runtime of \cref{main-nonsparse} to other parameters.
\begin{fact}\label{eqn:xstarLower}
$\|x^*\| \geq \frac{\|A\|}{\|A\|^2 + \lambda}\|\Pi_{A,\lambda}b\|$ when $\|A\| \geq \sqrt{\lambda}$, and $\|x^*\| = \frac{1}{2\sqrt{\lambda}}\|\Pi_{A,\lambda}b\|$ otherwise.
\begin{align}
    \|x^*\| &= \|(A^\dagger A + \lambda I)^+ A^\dagger b\| \nonumber\\
    &= \|(A^\dagger A + \lambda I)^+ A^\dagger\Pi_{A,\lambda}^+\Pi_{A,\lambda} b\| \tag*{since $A^\dagger = A^\dagger \Pi_{A,\lambda}^+ \Pi_{A,\lambda}$}\nonumber\\
    &\geq \|\Pi_{A,\lambda} b\|\min_{\substack{v \in \operatorname{span}(A^\dagger) \\ \|v\| = 1}} \|(A^\dagger A + \lambda I)^+ A^\dagger \Pi_{A,\lambda}^+ v\|\label{eq:projectedIncresae}
\end{align}
The $\min$ term is equal to the minimum non-zero singular value of $(A^\dagger A + \lambda I)^+ A^\dagger \Pi_{A,\lambda}^+$.
The non-zero singular values of this expression are $g(\sigma_i)$, where $\sigma_i$ is a non-zero singular value of $A$ and
\begin{align*}
    g(\sigma) = \frac{\sigma}{(\sigma^2+\lambda)(p_{A,\lambda}(\sigma))}
    = \begin{cases}
        \frac{1}{2\sqrt{\lambda}} & 0 < \sigma \leq \sqrt{\lambda} \\
        \frac{\sigma}{\sigma^2+\lambda} & \sqrt{\lambda} < \sigma.
    \end{cases}
\end{align*}
    This function is non-increasing, so the minimum singular value of $(A^\dagger A + \lambda I)^+ A^\dagger \Pi_{A,\lambda}^+$ is $g(\|A\|)$, proving the above fact by observing that for $\|A\| \leq \sqrt{\lambda}$ the operator $(A^\dagger A + \lambda I)^+ A^\dagger \Pi_{A,\lambda}^+$ is proportional to a projector, thus \cref{eq:projectedIncresae} becomes an equality.
\end{fact}
\section*{Acknowledgments}
    E.T.\ thanks Kevin Tian immensely for discussions integral to these results.
    Z.S.\ thanks Runzhou Tao and Ruizhe Zhang for giving helpful comments on the draft.
    
    A.G.\ acknowledges funding provided by Samsung Electronics Co., Ltd., for the project ``The Computational Power of Sampling on Quantum Computers'', and additional support by the Institute for Quantum Information and Matter, an NSF Physics Frontiers Center (NSF Grant PHY-1733907), as well as support by ERC Consolidator Grant QPROGRESS, by QuantERA project QuantAlgo 680-91-034 and also by the EU's Horizon 2020 Marie Skłodowska-Curie program 891889-QuantOrder.
    Z.S.\ was partially supported by Schmidt Foundation and Simons Foundation.
    E.T.\ is supported by the National Science Foundation Graduate Research Fellowship Program under Grant No. DGE-1762114.

\bibliographystyle{quantum}
\bibliography{ref}

\appendix
\section{Stochastic gradient bounds}

\grad*

\accvariance*

\begin{proof}
\hfill
\paragraph{Part 1.} First, we show $\nabla g(x)$ is unbiased.
\begin{align*}
    \E[ \nabla g(x) ] 
    = & ~ \E\Big[\frac{1}{C}\sum_{j=1}^C \frac{\nrm{A}_\fr^2}{\nrm{A_{r,*}}^2}\frac{\|A_{r,*}\|^2}{|A_{r,c_j}|^2} A_{r,c_j}x_{c_j} (A_{r,*})^\dagger\Big] - A^\dagger b + \lambda x
    \tag*{by definition of $\nabla g$} \\
    = & ~ \E_{r, c_1} \Big[\frac{\|A\|_\fr^2}{|A_{r,c_1}|^2} A_{r,c_1}x_{c_1} (A_{r,*})^\dagger\Big] - A^\dagger b + \lambda x
    \tag*{since the $c_i$'s are i.i.d.} \\
    = & ~ \Big(\sum_{r}\sum_{c_1} \frac{\|A_{r,c_1}\|^2 }{\|A\|_\fr^2} \frac{\|A\|_\fr^2}{\|A_{r,c_1}\|^2}  A_{r,c_1}x_{c_1} (A_{r,*})^\dagger\Big) - A^\dagger b + \lambda x \\
    = & ~ A^\dagger Ax - A^\dagger b + \lambda x = \nabla f(x).
\end{align*}

\paragraph{Part 2.}
Now variance.
We use that, for a random vector $v$, $\Var[v] = \E[\nrm{v - \E[v]}^2] = \E[\nrm{v}^2] - \nrm{\E[v]}^2$.
\begin{align}
    \Var[  \nabla g(x) ]
    & = \Var[  \nabla g(x)  + A^\dagger b - \lambda x ]\nonumber \\
    & = \E[\|\nabla g(x) + A^\dagger b - \lambda x - \E[\nabla g(x) + A^\dagger b - \lambda x]\|^2]
    \tag*{by definition of $\Var[\cdot]$}\nonumber \\
    &= \Var\Big[\frac{\|A\|_\fr^2}{\|A_{r,*}\|^2}\Big(\frac1C\sum_{j=1}^C\frac{\|A_{r,*}\|^2}{\abs{A_{r,c_{j}}}^2} A_{r,c_{j}}x_{c_{j}}\Big) (A_{r,*})^\dagger\Big] 
    \tag*{by definition of $\nabla g$}\nonumber\\
    &= \Var\Big[\Big(\frac1C\sum_{j=1}^C\frac{\|A\|_\fr^2}{\abs{A_{r,c_{j}}}^2} A_{r,c_{j}}x_{c_{j}}\Big) (A_{r,*})^\dagger\Big]\nonumber \\
    &= \E\Big[\norm[\Big]{\Big(\frac1C\sum_{j=1}^C\frac{\|A\|_\fr^2}{\abs{A_{r,c_{j}}}^2} A_{r,c_{j}}x_{c_{j}}\Big) (A_{r,*})^\dagger}^2\Big] - \norm{A^\dagger Ax}^2 .\nonumber \\
    &= \underbrace{\E\Big[\abs[\Big]{\frac1C\sum_{j=1}^C\frac{\|A\|_\fr^2}{\abs{A_{r,c_{j}}}^2} A_{r,c_{j}}x_{c_{j}}}^2 \norm{A_{r,*}}^2\Big]}_V - \norm{A^\dagger Ax}^2 . \label{eqn:v-def}
\end{align}
We expand the first term, using that it is an average of i.i.d.\ random variables:
\begin{align*}
    V
    &= \sum_{i=1}^m \frac{\|A_{i,*}\|^2}{\|A\|_\fr^2} \E_{c_1,\ldots,c_C}\Big[\abs[\Big]{\frac1C\sum_{j=1}^C\frac{\|A\|_\fr^2}{|A_{i,c_{j}}|^2} A_{i,c_{j}}x_{c_{j}}}^2\Big] \norm{A_{i,*}}^2 \\
    &= \|A\|_\fr^2\sum_{i=1}^m \E_{c_1,\ldots,c_C}\Big[\abs[\Big]{\frac1C\sum_{j=1}^C\frac{\|A_{i,*}\|^2}{|A_{i,c_{j}}|^2} A_{i,c_{j}}x_{c_{j}}}^2\Big]  \\
    &= \|A\|_\fr^2\sum_{i=1}^m \Big(\frac{1}{C}\Var_{c_1}\Big[\frac{\|A_{i,*}\|^2}{|A_{i,c_{1}}|^2} A_{i,c_{1}}x_{c_{1}}\Big] + \abs{A_{i,*}x}^2\Big)  \\
    &= \|A\|_\fr^2\sum_{i=1}^m \Big(\frac{1}{C}\sum_{j = 1}^n\Big(\frac{\abs{A_{i,j}}^2}{\|A_{i,*}\|^2}\frac{\|A_{i,*}\|^4}{|A_{i,j}|^4} \abs{A_{i,j}x_{j}}^2\Big) - \frac{1}{C}\abs{A_{i,*}x}^2 + \abs{A_{i,*}x}^2\Big)  \\
    &= \|A\|_\fr^2\sum_{i=1}^m \Big(\frac{1}{C}\|A_{i,*}\|^2\|x\|^2 + \Big(1-\frac{1}{C}\Big)(A_{i,*}x)^2\Big)  \\
    &= \frac{1}{C}\|A\|_\fr^4\|x\|^2 + \Big(1-\frac{1}{C}\Big)\|A\|_\fr^2\|Ax\|^2
\end{align*}

\paragraph{Part 3.}
Finally, the expression for $\E[\|\nabla g(x) - \nabla g(y)\|^2]$ follows from the observation that $\nabla g(x) - \nabla g(y)$ is simply the expression for $\nabla g(x-y)$, taking $b$ to be the zero vector.
\end{proof}

\variancevariant*

\begin{proof} \hfill
\paragraph{Part 1.}
We first show the expectation,
\begin{align*}
    \E[\nabla \tilde{g}(v)] &= \E\Big[\frac{\|A\|_\fr^2}{\|A_{r,*}\|^2}\Big(\frac1C\sum_{j=1}^C\frac{\|A_{r,*}\|^2}{|A_{r,c_j}|^2} A_{r,c_j}(A_{*,c_j})^\dagger v\Big)e_r - b + \lambda v\Big] \\
    &= \sum_{r=1}^m\E_{c_1,\ldots,c_C}\Big[\frac1C\sum_{j=1}^C\frac{\|A_{r,*}\|^2}{|A_{r,c_j}|^2} A_{r,c_j}(A_{*,c_j})^\dagger v\Big]e_r - b + \lambda v \\
    &= \sum_{r=1}^m\sum_{c=1}^n A_{r,c}((A_{*,c})^\dagger v) e_r - b + \lambda v \\
    &= AA^\dagger v - b + \lambda v.
\end{align*}
\paragraph{Part 2.}
For the variance in the $D$ norm, we reuse a computation from the proof of \cref{lem:acc-variance}:
\begin{align*}
    \E[\|\nabla \tilde{g}(v) - \E[\nabla \tilde{g}(v)]\|_D^2]
    &= \Var\Big[D\frac{\|A\|_\fr^2}{\|A_{r,*}\|^2}\Big(\frac1C\sum_{j=1}^C\frac{\|A_{r,*}\|^2}{|A_{r,c_j}|^2} A_{r,c_j}(A_{*,c_j})^\dagger v\Big) e_r\Big] \\
    &= \underbrace{\E\Big[\|A_{r,*}\|^2\abs[\Big]{\frac1C\sum_{j=1}^C\frac{\|A\|_\fr^2}{\abs{A_{r,c_j}}^2} A_{r,c_j}(A_{*,c_j})^\dagger v}^2\Big]}_{V \text{ from \cref{eqn:v-def}, where $x = A^\dagger v$}} - \norm{AA^\dagger v}_D^2 \\
    &= \frac{1}{C}\|A\|_\fr^4\|A^\dagger v\|^2 + \Big(1-\frac{1}{C}\Big)\|A\|_\fr^2\|AA^\dagger v\|^2 - \norm{AA^\dagger v}_D^2 \\
    &\leq 2\|A\|_\fr^2\|A\|^2\|A^\dagger v\|^2. \tag*{\qedhere}
\end{align*}
\end{proof}

\end{document}